\documentclass[12pt]{llncs}

\newcommand{\Mp}{M_{\phi}}

\begin{document}

\title{Hardness Results for the Gapped Consecutive-Ones Property
  Problem}

\author{Cedric Chauve\inst{1}, J{\' a}n Ma{\v n}uch\inst{1,2} and Murray Patterson\inst{2}}

\institute{Department of Mathematics, Simon Fraser University,
  Burnaby, BC, Canada \and School of Computing Science, Simon Fraser
  University, Burnaby, BC, Canada}

\maketitle


\begin{abstract}
Motivated by problems of comparative genomics and paleogenomics,
in~\cite{chauve-on} the authors introduced the Gapped Consecutive-Ones
Property Problem $(k,\delta)$-C1P: given a binary matrix $M$ and two
integers $k$ and $\delta$, can the columns of $M$ be permuted such
that each row contains at most $k$ blocks of ones and no two
consecutive blocks of ones are separated by a gap of more than
$\delta$ zeros.  The classical C1P problem, which is known to be
polynomial is equivalent to the $(1,0)$-C1P problem.  They showed that
the $(2,\delta)$-C1P Problem is NP-complete for all $\delta\geq 2$ and
that the $(3,1)$-C1P problem is NP-complete.  They also conjectured
that the $(k,\delta)$-C1P Problem is NP-complete for $k\geq 2$,
$\delta\geq 1$ and $(k,\delta)\neq (2,1)$.  Here, we prove that this
conjecture is true.  The only remaining case is the $(2,1)$-C1P
Problem, which could be polynomial-time solvable.
\end{abstract}

\section{Introduction} \label{sec-intro}

Let $M$ be a binary matrix with $n$ rows and $m$ columns. A {\em
  block\/} in a row of $n$ is a maximal sequence of consecutive
entries containing $1$. A {\em gap\/} is a sequence of consecutive
zeros that separates two blocks; the size of a gap is the length of
the sequence of zeros. $M$ is said to have the Consecutive-Ones
Property (C1P) if its columns can be permuted such that each row
contains one block (no gap then).  We call a permutation of the
columns of $M$ that witnesses this property a {\em consecutive-ones
  ordering\/} of $M$, and the resulting matrix of such a permutation
is {\em consecutive\/}.  Testing a binary matrix $M$ for the C1P can
be done in linear time~\cite{booth-testing,mcconnell-certifying}.
Matrix $M$ has the C1P if and only if a $PQ$-tree~\cite{booth-testing}
can be built for $M$, moreover, the $PQ$-tree stores all
consecutive-ones orderings of $M$.  The C1P has also been used in
molecular biology, in relation with physical
mapping~\cite{alizadeh-physical} and the reconstruction of ancestral
genomes~\cite{chauve-methodological} as follows: each column of the
matrix represents a genomic marker (sequence) that is believed to have
been present (up to small evolutionary changes such as nucleotide
mutations or small rearrangements) and unique in the considered
ancestral genome or physical map, and each row of the matrix
represents a set of markers that are believed to have been contiguous
along an ancestral chromosome, and the goal is to find one (or several
if possible) total orders on the markers that respect all rows (i.e.,
that keep all entries 1 consecutive in each
row). See~\cite{chauve-methodological} for a comprehensive
introduction to this problem. However, a common problem in such
applications is that matrices obtained from experiments do not have
the C1P~\cite{goldberg-four,chauve-methodological}.

Handling a matrix $M$ that does not have the C1P has been approached
using different points of view.  A first general approach consists of
transforming $M$ into a matrix that has the C1P, while minimizing the
modifications to $M$; such modifications can involve either in
removing rows, or columns, or both, or in flipping some entries from 0
to 1 or 1 to 0. In all cases, the corresponding optimization problems
have been proven NP-hard~\cite{dom-approximability,haj-note}. A second
approach consists of relaxing the condition of consecutivity of the
ones of each row, by allowing gaps, with some restriction to these
gaps. The question is then to decide if there is an ordering of the
columns of $M$ that satisfies these relaxed C1P conditions. As far as
we know, the only restriction that has been considered is the number
of gaps, either per row or in $M$.  In~\cite{goldberg-four}, the
authors introduced the notion of the $k$-{\em consecutive-ones
  property\/} ($k$-C1P).  A binary matrix $M$ has the $k$-C1P when its
set of columns can be permuted such that each row contains at most $k$
blocks.  They call a permutation of the columns of $M$ that witnesses
this property a $k$-{\em consecutive-ones ordering\/} of $M$, and the
resulting matrix of such a permutation is $k$-{\em consecutive\/}.
In~\cite{goldberg-four}, the authors show that deciding if a binary
matrix $M$ has the $k$-C1P is NP-complete, even if $k=2$.  Also,
finding an ordering of the columns that minimizes the number of gaps
in $M$ is NP-complete even if each row of $M$ has at most two
ones~\cite{haddadi-note}.

In the present work, we follow the second approach, motivated by the
problem of reconstructing ancestral genomes using max-gap
clusters~\cite{chauve-methodological}: the restrictions to the allowed
gaps are that both the number of gaps per row and the size of each gap
are bounded.  Formally, let $k$ and $\delta$ be two integers. A binary
matrix $M$ is said to have the $(k,\delta)$-Consecutive-Ones Property,
denoted by $(k,\delta)$-C1P, if its columns can be permuted such that
each row contains at most $k$ blocks and no gap larger than $\delta$.
Here, we call a permutation of the columns of $M$ that witnesses this
property a $(k,\delta)$-{\em consecutive-ones ordering\/} of $M$, and
the resulting matrix of such a permutation is $(k,\delta)$-{\em
  consecutive\/}.  In~\cite{chauve-on}, we introduced this problem and
gave preliminary complexity and algorithmic results. In particular we
showed that the $(2,\delta)$-C1P Problem is NP-complete for all
$\delta\geq 2$ and that the $(3,1)$-C1P problem is NP-complete. In the
present work, we settle the complexity for all possible values of $k$
and $\delta$: we show that testing for the $(k,\delta)$-C1P is
NP-complete for every $k\geq2,\delta\geq1$, $(k,\delta)\neq
(2,1)$. This leaves only one case open: the $(2,1)$-C1P Problem. Note
that from an application point of view (i.e., paleogenomics and the
reconstruction of ancestral genomes), answering the ($k$,$\delta$)-C1P
Problem for small values of both $k$ and $\delta$ is very
relevant. Indeed, in most cases, it is errors in computing the initial
matrix $M$ that makes it not have the C1P: these errors correspond to
small gaps in some rows of this matrix. These errors are due to small
overlapping genome rearrangements or mistakes in identifying proper
ancestral genomic markers.

In Section~\ref{sec-notation}, we introduce notations related to the
gapped-C1P problem. Then, in Section~\ref{sec-results}, we state and
prove our two main results. The main point in our proofs is a more
general result that states that, given an arbitrary binary matrix, one
can add a relatively small number of additional rows to the matrix
such that the order of a chosen subset of columns must be fixed if
some gaps conditions among these columns are to be respected.  We
believe this result can have applications in other problems related to
the C1P. Finally, we conclude with some open problems and
perspectives.

\section{Notation and Conventions} \label{sec-notation}

First, we introduce some notation and conventions that we use in the
following.  We have the binary matrix $M$ on the set $\{1,\dots,N\}$
of columns.  In the constructions used to show NP-completeness, we will
divide columns of the matrix into ordered sequences of blocks
$b_1,\dots,b_m$ by designing rows enforcing the columns of each block
to appear consecutive and the blocks to appear in the order
$b_1,\dots,b_m$ (or in the reversed order), i.e., for any $i<j$,
column $c \in b_i$ and $d \in b_j$, $c$ appears before $d$ in any
$(k,\delta)$-consecutive ordering of $M$ for any $k\geq 2,\delta\geq
1$.  Furthermore, the columns of a block $b_i$ will be denoted
$b_i^1,\dots,b_i^{|b_i|}$.

To specify a row in the matrix $M$, we use the convention of only
listing in the square brackets, the columns that contain $1$ in this
row.  For example, $[1,5,8]$ represents a row with ones in columns
$1$, $5$ and $8$, and zeroes everywhere else.  We will also use blocks
to specify columns in the block, for example, if $b_1 = \{1,2,3\}$,
then $[b_1,5]$ would mean $[1,2,3,5]$ and $[b_1 \setminus
\{b_1^2\},4,5]$ would mean $[1,3,4,5]$.

Given a column $i$ in matrix $M$ and an integer $d \geq 0$, the set of
columns $N_d(i) = \{i-d,\dots,i-1,i+1,\dots,i+d\}$ of $M$ is called
the $d$-{\em neighborhood\/} of $i$.

\section{Results} \label{sec-results}

First, we have the following important property of matrices which have
the $(k,\delta)$-C1P, for every $k \geq 2$, $\delta \geq 1$.

\begin{theorem} \label{thm-fix-order}
For all $k \geq 2, \delta \geq 1$ and $n \geq 2\delta + 3$, given
matrix $M$ on $N \geq n$ columns, $n(\delta+1)-\frac{\delta (\delta
+3)}{2} -1$ rows can be added to $M$ to force $n$ selected columns to
appear consecutive and in fixed order (or the reverse order) in any
$(k,\delta)$-consecutive ordering of $M$.
\end{theorem}

\begin{proof}
Given that $1,\dots,N$ are the columns of $M$, let $C =
\{i+1,i+2,\dots,i+n\}$, for some $i \leq N-n$ be the subset of $n$
columns that we want to force to appear consecutive and in this order
(or the reverse order) in any $(k,\delta)$-consecutive ordering of $M$
for any $k\geq 2,\delta \geq 1$.  Throughout the proof, when the
context is clear that we are referring only to the elements of $C$, we
denote $C = \{1,\dots,n\}$, and index its elements accordingly.

We add the rows $[i,j]$ to $M$, for any $1\leq i<j\leq n$ such that
$|i-j| \leq \delta+1$.  This amounts to adding $(n-(\delta +
1))(\delta +1) + \delta + (\delta -1) +\cdots +2+1 = (n-(\delta
+1))(\delta + 1) + \frac{\delta (\delta +1)}{2} = n\delta + n -
\frac{\delta^2}{2} - \frac{3\delta}{2}-1 =
n(\delta+1)-\frac{\delta(\delta+3)}{2}-1$ rows to $M$.  We now show
that the columns in $C$ appear in the sequence $1,\dots,n$, or
$n,\dots,1$ in any $(k,\delta)$-consecutive ordering of $M$.  If we
represent any $(k,\delta)$-consecutive ordering of $M$ by a
permutation $\pi$ of the columns of $M$, i.e., $\pi(i)$ is the $i$-th
column in the permuted matrix, $\pi(M)$ is the entire permuted matrix,
then we have the following claim.

\begin{claim} \label{clm-neighborhood}
For any $\pi(i),\pi(j) \in C$, if $|\pi(i)-\pi(j)| \leq \delta+1$ then
$|i-j| \leq \delta+1$.
\end{claim}

\begin{proof}
If $1\leq \pi(i),\pi(j) \leq n$ and $|\pi(i)-\pi(j)| \leq \delta+1$
then $M$ contains a row $[\pi(i),\pi(j)]$.  Hence, in the permuted
matrix, $\pi(M)$, we have a row $[i,j]$.  Since $\pi(M)$ is a
$(k,\delta)$-consecutive ordering of $M$, there can be at most
$\delta$ zeros between columns $i$ and $j$ in $\pi(M)$, and hence
$|i-j| \geq \delta+1$.
\end{proof}

Note that another way of stating this claim is: For any $\pi(i),\pi(j)
\in C$, if $\pi(j) \in N_{\delta+1}(\pi(i))$ then $j \in
N_{\delta+1}(i)$.

Next, we will show that the columns in $C$ have to appear consecutive
in any $(k,\delta)$-consecutive ordering of $M$.  Let $i_{\min}$
$(i_{\max})$ be the first (last) column in $\pi(M)$ containing a
column in $C$, i.e., $i_{\min} = \min_{c\in C}\pi^{-1}(c)$ and
$i_{\max} = \max_{c\in C}\pi^{-1}(c)$.  Then this consecutiveness
property can be expressed as follows.

\begin{claim} \label{clm-consec}
We have that $i_{\max}-i_{\min}=m-1$.
\end{claim}

\begin{proof}
Consider an $i \in M$ such that $\pi(i)$ is in the middle part of $C$,
in $C_{\mbox{MID}}=\{\delta+1,\dots,m-\delta-1\} \neq \emptyset$
($C_{\mbox{MID}} \neq \emptyset$ since $n \geq 2\delta+3$).
Obviously, $i_{\min} \leq i \leq i_{\max}$.  Then, for every $d \in
N_{\delta+1}(\pi(i))$, $d \in C$, and hence, $i_{\min} \leq
\pi^{-1}(d) \leq i_{\max}$, and by the first claim, also $\pi^{-1}(d)
\in N_{\delta+1}(i)$.  Since permutation $\pi$ is a one-to-one mapping
from the set $M$ to itself, and $|N_{\delta+1}(\pi(i))|$ is
$2\delta+2$ ($|N_{\delta+1}(i)|$ is $2\delta+2$), it follows that for
each $j$ such that $j \in N_{\delta+1}(i)$, there is a $d \in
N_{\delta+1}(\pi(i)) \subseteq C$ such that $\pi(j)=d$.  Hence, for
every $i$ such that $\pi(i) \in C_{\mbox{MID}}$, we have that for
every $j \in N_{\delta+1}(i)$, $\pi(j) \in C$.  Consequently, for
every such an $i$, $i \in I =
\{i_{\min}+\delta+1,\dots,i_{\max}-\delta-1\}$.

Let $i_1$ $(i_2)$ be the smallest (largest) $i$ such that $\pi(i) \in
C_{\mbox{MID}}$.  Recall that $i_1,i_2 \in I$.  Let $C_{\mbox{BOR}} =
C \setminus C_{\mbox{MID}}$.  Since, for all $j \in N_{\delta+1}(i_1)
\cup N_{\delta+1}(i_2)$, $\pi(j) \in C$, we have that
$\pi(i_1-\delta-1),\dots,\pi(i_1-1),\pi(i_2+1),\dots,\pi(i_2+\delta+1)
\in C_{\mbox{BOR}}$.  Note that these $2\delta+2$ elements in
$C_{\mbox{BOR}}$ are distinct, even if $i_1 = i_2$, the case that
arises when $n = 2\delta+3$.  By the definitions of $C_{\mbox{MID}}$,
$i_{\min}$ and $i_{\max}$, it follows that $\pi(i_{\min})$ and
$\pi(i_{\max})$ are also in $C_{\mbox{BOR}}$.  Hence, if either $i_1 >
i_{\min}+\delta+1$ or $i_2 < i_{\max}-\delta-1$, then we have at least
$2\delta+3$ distinct values from $C_{\mbox{BOR}}$, which is a
contradiction, since by the fact that $n \geq 2\delta+3$, and by the
definition of $C_{\mbox{MID}}$, $|C_{\mbox{BOR}}| = 2\delta+2$.
Therefore, $i_1=i_{\min}+\delta+1$, $i_2=i_{\max}-\delta-1$, and for
all $i \in \{i_{\min},\dots,i_{\max}\} \setminus I$, $\pi(i) \in
C_{\mbox{BOR}}$.  Thus for all $i \in I$, either $\pi(i) \in
C_{\mbox{MID}}$ or $\pi(i) \not\in C$.

If there is no $i \in \{i_{\min},\dots,i_{\max}\}$ such that $\pi(i)
\not\in C$, then all the elements in
$\pi(i_{\min}),\dots,\pi(i_{\max})$ are in $C$, and the claim follows.
Assume there is an $i$ such that $i \not\in C$, and let $i_0$ be the
smallest such $i$.  Since, for all $i \in \{i_{\min},\dots,i_{\max}\}
\setminus I$, $\pi(i) \in C_{\mbox{BOR}} \subseteq C$, it follows that
that $i_0 \in I$, where $i_0 \neq i_1$, by the definition of $i_1$.
Therefore, $i_0 > i_1 = i_{\min}+\delta+1$, and hence, $\pi(i_0-1) \in
C_{\mbox{MID}}$.  Since $i_0 \in N_{\delta+1}(i_0-1)$, it follows that
$i_0$ must also be in $C$, contradicting this assumption, thus the
claim follows.
\end{proof}

Now, by the previous claim, we have that the set of columns $C
\subseteq M$ are consecutive in any $(k,\delta)$-consecutive ordering
of $M$.  Given this, and the fact that any column of $M \setminus C$
is zero in any of these rows added to $M$ to force the columns of $C$
to be consecutive, this set of rows is $(k,\delta)$-consecutive for
any permutation of the columns of $M$, provided only that the columns
$C$ are consecutive somewhere in this ordering of $M$.  Hence, to
prove the theorem, it is sufficient to show that in the case that $M =
C = \{1,\dots,n\}$, the columns of $\pi(M)$ are ordered either in
increasing or decreasing order in any $(k,\delta)$-consecutive
ordering of $M$.

We will proceed by induction on $n$.  We need the following claim.

\begin{claim} \label{clm-order}
If $M = C$, then either for all $i \in
\{1,\dots,\delta+1,n-\delta,\dots,n\}$, $\pi(i)=i$ or for all $i \in
\{1,\dots,\delta+1,n-\delta,\dots,n\}$, $\pi(i)=n-i+1$.
\end{claim}

\begin{proof}
We will show the claim by induction on $i$.  In the base case, we need
to show that $\{\pi(1),\pi(n)\}=\{1,n\}$.  Assume that both $\pi(1)$
and $\pi(n)\not\in\{1,n\}$.  Then the set $N_{\delta+1}(\pi(1)) \cap
M$ has more that $\delta+1$ elements.  By the first claim, for every
$d \in N_{\delta+1}(\pi(1)) \cap M$, $\pi^{-1}(d) \in N_{\delta+1}(1)
\cap M$.  Since $\pi$ is a one-to-one mapping from the set $M$ to
itself, and $|N_{\delta+1}(\pi(1)) \cap M| > \delta+1$, then this
implies that $|N_{\delta+1}(1) \cap M| > \delta+1$.  This is a
contradiction, because $|N_{\delta+1}(1)| = \delta+1$.  Hence, either
$\pi(1) = 1$ or $\pi(1) = n$, and similarly, $\pi(n)=1$ or
$\pi(n)=n$.  Without loss of generality, we can assume that $\pi(1)=1$
and $\pi(n)=n$, and show by induction that the columns in $\pi(M)$ are
ordered in increasing order.

For the inductive step, consider an $i\leq \delta+1$ and assume that
$\pi(j)=j$ for every $j\in\{1,\dots,i-1,n-i+2,\dots,n\}$.  By the
induction hypothesis, $\pi(i)\in\{i,\dots,n-i+1\}$.  Assume that
$\pi(i)>i$ and $\pi(i)<n-i+1$.  Then the set $N_{\delta+1}(\pi(i))
\cap M$ has more than $\delta+i$ elements.  Again, by the first claim,
and the fact that $\pi$ is a one-to-one mapping, this implies that
$|N_{\delta+1(i)} \cap M| > \delta+i$, a contradiction.  Hence, either
$\pi(i)=i$ or $\pi(i)=n-i+1$.  Assume that $\pi(i)=n-i+1$.  By the
induction hypothesis, $\pi(n)=n$.  Obviously, then
$|\pi(n)-\pi(i)|=|n-(n-i+1)|=i-1\leq\delta+1$, and hence, by the first
claim, $|n-i|\leq\delta+1$.  Since $n \geq 2\delta+3$, and
$i\leq\delta+1$, then $|n-i|=n-i\geq 2\delta+3-(\delta+1)=\delta+2$,
which is a contradiction.  Thus, $\pi(i)=i$, and similarly,
$\pi(n-i+1)=n-i+1$.
\end{proof}

We now proceed by induction on $n$, to prove the theorem.  For the
base case, assume that $n=2\delta+3$.  By the last claim, for every
$i\in M\setminus\{\delta+2\}$, $\pi(i)=i$ ($\pi(i)=n-i+1$,
respectively).  It then follows, by the fact that $\pi$ is a
one-to-one mapping from the set $M$ to itself, that
$\pi(\delta+2)=\delta+2$.

Now, for induction, assume that $n>2\delta+3$.  Since $\delta\geq 1$,
by the last claim, either $\pi(1)=1$, $\pi(2)=2$ or $\pi(1)=m$,
$\pi(2)=m-1$.  Without loss of generality, assume that $\pi(1)=1$ and
$\pi(2)=2$.  Consider $M'$, the matrix that results from the removal
of column $1$ from $M$, and all rows $[1,i]$, for $i=2,\dots n$, from
this set of rows we add to $M$.  By the induction hypothesis, $M'$ is
$(k,\delta)$-consecutive, $k\geq 2,\delta\geq 1$, only for the orders
$\{2,\dots,n\}$ and $\{n,\dots,2\}$ of the columns of $M'$.  So if the
columns $M\setminus \{1\}$ are ordered $\{2,\dots,n\}$, since
$\pi(1)=1$, then the theorem holds.  Otherwise, the columns
$M\setminus \{1\}$ are ordered $\{n,\dots,2\}$, and thus $\pi(2)=m$,
which is a contradiction.  Thus the theorem holds.
\end{proof}

We now use this Theorem~\ref{thm-fix-order} to construct a reduction
from 3SAT to the problem of testing for the $(k,\delta)$-C1P to show
that this problem is NP-complete for every $k,\delta\geq 2$.

\begin{theorem} \label{thm-kdNPc}
Testing for the $(k,\delta)$-C1P is NP-complete for every
$k,\delta\geq 2$.
\end{theorem}

\begin{proof}
Let $\phi$ be a 3CNF formula over the $n$ variables
$\{v_1,\dots,v_n\}$, with $m$ clauses $\{C_1,\dots,C_m\}$.  We
construct a matrix $\Mp$ with $2n+d+5m$ columns and $n+6m+2d-3$ rows,
where $d=\max\{2k,5\}$, such that $\Mp$ has the $(k,\delta)$-C1P iff
$\phi$ is satisfiable for $k,\delta \geq 2$.

In~\cite{goldberg-four}, the authors show that, given a 3CNF formula
$\phi$, they can construct a matrix $\Mp$ that has the $k$-C1P iff
$\phi$ is satisfiable for $k\geq 2$.  Our construction is very similar
to this, with the extra condition that $\Mp$ cannot have any gap
larger than $\delta$.

To achieve this, we first force a subset of the columns of $M_{\phi}$
to be consecutive and in fixed order in any $(k,1)$-consecutive
ordering of $\Mp$, and then we will build off of this, a construction
similar to that of~\cite{goldberg-four}.  In particular, we impose
this order on the subset $\{2n+1,\dots,2n+d\}$ of the columns
$\{1,\dots,2n+d+5m\}$ of $\Mp$ by adding the
$d(\delta+1)-\frac{\delta(\delta+3)}{2}-1=2d-3$ rows $[i,j]$ to $\Mp$,
for any $2n+1\leq i<j\leq 2n+d$ such that $|i-j|\leq\delta+1$.  By
Theorem~\ref{thm-fix-order}, these $d$ columns must be in fixed order
(or the reverse).  We can assume the former without loss of
generality.

Now we associate variable $v_i$ with block $b_i=\{2i-1,2i\}$, for
$i=1,\dots,n$, imposing the same restrictions on these columns as
in~\cite{goldberg-four}.  So for each $b_i$, we add the row
$[b_i,b_{i+1},\dots,b_n,2n+1,2n+3,\dots,2n+2k-3,2n+2k-1]$ to $\Mp$.

Next we associate clause $C_j$ with block
$B_j=\{2n+d+5j-4,\dots,2n+d+5j\}$, for $j=1,\dots,m$, and add the row
$[2n+d-2k+2,2n+d-2k+4,\dots,2n+d-4,2n+d-2,2n+d,B_1,B_2,\dots,B_j]$ to
$\Mp$.

Now the columns of every $(k,\delta)$-consecutive ordering of the
matrix $\Mp$ are ordered: the blocks $b_1,\dots,b_n$, followed by the
$d$ columns $2n+1,\dots,2n+d$ that remain consecutive and in order,
followed by blocks $B_1,\dots,B_m$.  We now add the same rows to $\Mp$
as in~\cite{goldberg-four} to associate each clause to its three
variables to properly simulate 3SAT, only that within the segment of
$d$ columns $2n+1,\dots,2n+d$, each row takes value
$[2n+2k-5,2n+2k-3,2n+2k-2,\dots,2n+d]$.  The idea is that this segment
of $d$ columns enforces $k-2$ gaps, while each gap is of size $1$.
\end{proof}

Finally, we slightly modify the construction in the proof of
Theorem~\ref{thm-kdNPc}, to show that testing for the $(k,1)$-C1P is
NP-complete for every $k\geq 3$ by reduction from 3SAT.

\begin{theorem} \label{thm:k1NPc}
Testing for the $(k,1)$-C1P is NP-complete for every $k\geq 3$.
\end{theorem}

\begin{proof}
Let $\phi$ be a 3CNF formula over the $n$ variables
$\{v_1,\dots,v_n\}$, with $m$ clauses $\{C_1,\dots,C_m\}$.  We
construct a matrix $\Mp$ with $2n+d+4m$ columns and $n+4m+2d-3$ rows,
where $d=\{2k,5\}$, such that $\Mp$ has the $(k,1)$-C1P iff $\phi$ is
satisfiable for $k\geq 3$.  We do this as follows.

We again associate columns $1,\dots,2n$ with the variables of $\phi$,
and again use Theorem~\ref{thm-fix-order} to force the subset
$\{2n+1,\dots,2n+d\}$ of the columns $\{1,\dots,2n+d+4m\}$ of $\Mp$ to
appear consecutive and in fixed order in any $(k,1)$-consecutive
ordering of $\Mp$ for $k\geq 2$.

We associate each clause $C_j \in \{C_1,\dots,C_m\}$, with block $B_j
= \{2n+d+4j-4,\dots,2n+d+4j\}$.  Now, we need to introduce only three
more rows to associate the clauses to their variables to properly
simulate 3SAT.  Suppose that clause $C_j$ contains the literal
$v_{\alpha}$.  As such, we add the row
$[2\alpha,2\alpha+1,\dots,2n+1,2n+3,2n+5,\dots,2n+2k-5,2n+2k-3,2n+2k-2,2n+d,B_j^1,B_j^2]$
to $M_{\phi}$.  If $v_{\alpha}$ is false, this forces $B_j^1$ and
$B_j^2$ to be among the first three columns of block $B_j$ in any
$(k,1)$-consecutive ordering of $\Mp$ for $k\geq 3$.  Note that any
other ordering of the columns of $B_j$ would introduce either a gap of
size $2$, or a $k$-th gap in this row.  If another literal in $C_j$ is
$v_{\beta}$, we add the row
$[2\beta,2\beta+1,\dots,2n+1,2n+3,2n+5,\dots,2n+2k-5,2n+2k-3,2n+2k-2,2n+d,B_j^1,B_j^3]$
to $\Mp$.  If $v_{\beta}$ is false, this forces $B_j^1$ and $B_j^3$ to
be among the first three columns of block $B_j$ in any
$(k,1)$-consecutive ordering of $\Mp$ for $k\geq 3$.  If $v_{\gamma}$
is the third literal of $C_j$, we add the row
$[2\gamma,2\gamma+1,\dots,2n+1,2n+3,2n+5,\dots,2n+2k-5,2n+2k-3,2n+2k-2,2n+d,B_j^1,B_j^4]$
to $\Mp$.  If $v_{\gamma}$ is false, this forces $B_j^1$ and $B_j^4$
to be among the first three columns of block $B_j$ in any
$(k,1)$-consecutive ordering of $\Mp$ for $k\geq 3$.  Finally, since
$B_j^1,B_j^2,B_j^3,B_j^4$ cannot simultaneously be among the first
three columns of block $B_j$, we have that not all three literals can
be false in any $(k,1)$-consecutive ordering of $\Mp$ for $k\geq 3$.
It is easy to show, that if any literal in $C_j$ is true, then there
is some $(k,\delta)$-consecutive ordering of the rows involving block
$B_j$.
\end{proof}

\section{Conclusion} \label{sec-concl}

While this work improves on the most interesting open question given
in~\cite{chauve-on}, there still remain several open questions.  The
remaining open question that is most interesting now is the complexity
of deciding the $(2,1)$-C1P for a binary matrix $M$.  Since the two
NP-completeness constructions presented here force either a gap of
size two, or at least two gaps of size one in any legal configuration
of $M$, if testing for the $(2,1)$-C1P is NP-complete, it would
certainly require a different type of construction.

Deciding the $k$-C1P, for $k \geq 2$ has been proven NP-complete
in~\cite{goldberg-four}, and we have shown that deciding the
$(k,\delta)$-C1P is NP-complete for $k \geq 2, \delta \geq 1,
(k,\delta) \neq (2,1)$.  However, the complexity of deciding the
gapped C1P when only $\delta$ is fixed (we call this the
$(*,\delta)$-C1P) is still an interesting open question.  We have a
preliminary proof that deciding the $(*,\delta)$-C1P is NP-complete
for all $\delta \geq 1$, by reducing from the version of 3SAT where
each variable appears at most twice positively and once negatively.

Another natural problem is the $(k,\delta)$-C1P Problem considered
here, but with a third parameter added, namely the maximum number of
entries $1$ that can be present in a row of $M$, called the {\em
  degree} of $M$. This problem is motivated by the fact that in the
framework described in~\cite{chauve-methodological}, it is possible to
constrain matrices used to reconstruct ancestral genomes to have a
small degree. Note that with matrices of degree $2$, the number of
gaps can be at most $1$, and the $(2,\delta)$-C1P problem is then
equivalent to the problem of deciding if the graph whose incidence
matrix is $M$ has bandwidth at most $(\delta+1)$. For $\delta=1$, the
graph bandwidth problem can be solved in linear
time~\cite{caprara-on}, while in~\cite{saxe-dynamic} a dynamic
programming algorithm with time and space complexity exponential in
$\delta$ was described. We adapted in~\cite{chauve-on} this algorithm
for testing the $(k,\delta)$-C1P for matrices of small degree, but the
exponential space complexity makes it difficult to use in practice on
matrices with degree greater than $3$.  However, deciding the
$(k,\delta)$-C1P for small values of $k$ and $\delta$ may become
tractable if the degree of the matrix is bounded as well.  The design
of efficient algorithms, both in time and space, for deciding the
gapped consecutive-ones property is a promising research avenue, with
immediate applications in genomics.

Adding the degree of the matrix as a third parameter (we call it $d$
here) to the problem of deciding the $(k,\delta)$-C1P to give the new
problem of deciding the $(d,k,\delta)$-C1P then introduces more
interesting open questions from a complexity theory perspective. We
know that deciding the $(d,k,\delta)$-C1P is polynomial-time solvable
by the above algorithm, and in fact, this problem where $k$ is
unbounded is just the $(d,d,\delta)$-C1P, because $k \leq d$.  The
complexity of deciding this property when $\delta$ is unbounded,
namely the $(d,k,*)$-C1P is still open.  We have a preliminary proof
that deciding the $(d,k,*)$-C1P, for all $d \geq 4, k \geq 3$ is
NP-Complete, by a reduction from 3SAT, leaving open the complexity of
deciding the $(4,2,*)$-C1P and the $(3,2,*)$-C1P.  While this implies
that this problem is intractible in general, in practice, $\delta$ and
$d$ are quite small, so the design of efficient algorithms for these
cases can still be a fruitful avenue of research.

From a purely combinatorial point of view, there has been a renewed
interest in the characterization of non-C1P matrices in terms of
forbidden submatrices introduced by Tucker~\cite{tucker-structure}. It
has recently been shown that this characterization could be used in
the design of algorithms related to the
C1P~\cite{dom-recognition,chauve-minimal}. The question there is the
following: is there a nice characterization of non $(k,\delta)$-C1P
matrices in terms of forbidden matrices?

Finally it is also natural to ask if there exists a structure that can
represent all orderings that satisfy some gaps conditions related to
the consecutive-ones property. Such a structure exists for the
ungapped C1P: for a matrix that has the C1P, its PQ-tree represents
all its valid consecutive orderings, and it can be computed in linear
time~\cite{mcconnell-certifying}. This notion has even been extended
to matrices that do not have the C1P through the notion of
PQR-tree~\cite{meidanis-on,mcconnell-certifying}. Although the
existence of such a structure with nice algorithmic properties is
ruled out by the hardness of deciding the gapped C1P, it remains open
to find classes of matrices such that deciding the gapped C1P is
tractable, and in such case, to represent all possible orderings in a
compact structure. Here again, this question is motivated both by
theoretical considerations (for example representing all possible
layouts of a graph of bandwidth $2$), but also by computational
genomics problems~\cite{chauve-methodological}.


\end{document}